\newcounter{myctr}
\def\myitem{\refstepcounter{myctr}\bibfont\noindent\ifnum\themyctr>9\else\phantom{0}\fi\hangindent17pt\themyctr.\enskip}

\documentclass{ws-ijqi}
\usepackage{hyperref}
\usepackage[super,sort,compress]{cite}
\begin{document}

%%%%%%%%%%%%%%%%%%%%% Publisher's Area please ignore %%%%%%%%%%%%%%
\catchline{}{}{}{}{}
%%%%%%%%%%%%%%%%%%%%%%%%%%%%%%%%%%%%%%%%%%%%%%%%%%%%%%%%%%%%%%%%%%%

\title{TIGHTER SUPERADDITIVITY RELATIONS FOR $l_{1}$-NORM COHERENCE MEASURE}

\author{Kang-Kang Yang}

\address{School of Mathematical Sciences, Capital Normal University, Beijing 100048, China\\
2220501004@cnu.edu.cn}

\author{Zhong-Xi Shen}

\address{School of Mathematical Sciences, Capital Normal University, Beijing 100048, China\\
18738951378@163.com}

\author{Zhi-Xi Wang}

\address{School of Mathematical Sciences, Capital Normal University, Beijing 100048, China\\
wangzhx@cnu.edu.cn}

\author{Shao-Ming Fei}

\address{School of Mathematical Sciences, Capital Normal University, Beijing 100048, China\\
feishm@cnu.edu.cn}

\maketitle

%\begin{history}
%\received{Day Month Year}
%\revised{Day Month Year}
%\accepted{Day Month Year}
%\comby{(xxxxxxxxxx)}
%\end{history}

\begin{abstract}
Quantum coherence serves as a crucial physical resource, with its quantification emerging as a focal point in contemporary research. Superadditivity constitutes one of the most fundamental attributes in characterizing the coherence distribution in multipartite quantum systems. In this paper, we provide a way to derive tighter superadditivity inequalities of $l_1$-norm coherence measure for arbitrary multiqubit states. We present a category of superadditivity relations related to the $\alpha$-th ($\alpha\geqslant 2$) power of $l_{1}$-norm coherence $C_{l_{1}}$ under certain conditions. Our results are better than existing ones and are illustrated in detail with examples.
\end{abstract}


%\tableofcontents  % optional

\markboth{Kang-Kang Yang, Zhong-Xi Shen, Zhi-Xi Wang and Shao-Ming Fei}
{TIGHTER SUPERADDITIVITY RELATIONS FOR $l_{1}$-NORM COHERENCE MEASURE}

\section{Introduction}	

Originated from the superposition principle of quantum mechanics, the quantum coherence is one of the most fundamental properties of quantum mechanics that is different from classical mechanics, and it is an important resource for completing various quantum information processing tasks such as quantum key distribution \cite{S1}, quantum metrology \cite{S2}, low temperature thermodynamics \cite{S3} and quantum biology \cite{S4}. With the development of quantum information science, the quantitative research of coherence resources has attracted extensive attention \cite{S5,S6,S7,S8,S9,S10}. In their seminal work in Ref.\citen{S5}, Baumgratz \emph{et~al.} proposed a frame of resource theory for quantifying coherence and demonstrated that both $l_{1}$-norm coherence and relative entropy coherence are bona fide coherence measures. In particular, the $l_{1}$-norm coherence characterizes the path information in quantifying the wave property in wave-particle duality \cite{S11,S12}, and is related to the speedup in quantum computing \cite{S13}.

In contrast to quantum correlation, which is defined for bipartite or multipartite systems, the quantum coherence depends on the reference bases of the whole system \cite{S14}. The interrelation between coherence and certain quantum correlations has been studied \cite{S15,S16,S17}. Analogous to the trade-off relations satisfied by the quantum entanglement and quantum discord \cite{S18,S19,S20}, the coherence contribution of parts in multipartite systems has also been studied \cite{S21,S22,S23}. The distribution of coherence in multipartite systems is described by the trade-off relations among the coherence of subsystems such as superadditivity \cite{S24,S25}. Superadditivity of coherence measures is an elementary property within the frame of resource theory of coherence.
Let $C$ be a bona fide coherence measure with respect to a reference basis $\left \{ |i\rangle\right \}$, the generalized correlated coherence for a multipartite quantum system is defined as \cite{S16}:
\begin{equation}
\begin{split}
\mathcal{S} = C(\rho_{A_{1}A_{2}\cdots A_{N}})
- \sum_{n=1}^{N} \lambda_{n}C(\rho_{A_{n}}), \label{1}
\end{split}
\end{equation}
where $C(\rho_{A_{1},\dots,A_{N}})$ denotes the total coherence of the $N$-partite quantum  states $\rho_{A_{1},\dots,A_{N}}$, $C(\rho_{A_{n}})$ is the local coherence of the reduced state $\rho_{A_{n}}$ with respect to the basis $\left\{|i\rangle_{A_{n}}\right \}$, the coefficients $\lambda_{n}$ are a set of real numbers satisfying $\lambda_{n}\geqslant 1$.
The superadditivity of $C$ is given by $\mathcal{S}\geqslant 0$. A smaller $\mathcal{S}$ indicates a stronger superadditivity relation.

Tan \emph{et~al.} has demonstrated that $l_{1}$-norm coherence are superadditive for bipartite quantum states \cite{S16}. However, the robustness of coherence measure does not satisfy superadditivity in the bipartite mixed states \cite{A26}. That means, superadditivity depends on the coherence measures adopted, and there is no consistent consensus among them\cite{A27}. When $\lambda_{n}=1$ for all $n$, the correlated coherence $\mathcal{S}$ of a multiqubit system is always positive for the $l_{1}$-norm coherence $C_{l_{1}}$ \cite{S25} and the relative entropy coherence $\mathcal{C}_{r}$ \cite{A28,S27,S14}. That is,
$C(\rho_{A_{1}A_{2}\cdots A_{N}})\geqslant C(\rho_{A_{1}})+C(\rho_{A_{2}})+\cdots+C(\rho_{A_{N}})$. Taking into account the distinct contributions of each local coherence, stronger superadditivity inequalities can be derived by adjusting the coefficients $\lambda_{n}$. Further, the superadditivity related to the $\alpha$-th power of coherence $C_{l_1}$ ($\alpha\geqslant 1$) in multiqubit systems have also been studied \cite{S29,S30,S31}.

In this paper, we first show that the superadditivity inequalities based on the $\alpha$-th ($\alpha\geqslant 1$) power of $C_{l_1}$ can be tighter. Then a category of tighter superadditivity inequalities ($\mathcal{S}\geqslant 0 $) related to the $\alpha$-th ($\alpha\geqslant 2$) power of $l_{1}$-norm coherence $C_{l_{1}}$ are presented for multiqubit systems. We conclude by giving detailed examples to illustrate that the new inequalities are better than all the existing ones.

\section{Preliminary on Superadditivity Relations of the $l_{1}$-norm of Coherence}

Consider a finite-dimensional Hilbert space $\mathcal{H}$ with dimension $d=\dim({\mathcal{H}})$. We fix the reference basis to be the computational basis $\{|i\rangle\}_{i=1}^{d}$. The $l_{1}$-norm of coherence $C_{l_{1}}$ of a quantum state $\rho$ is given by Ref.\citen{S5},
\begin{equation}
\begin{aligned}
C_{l_{1}}(\rho)=\sum\limits_{i\neq j}|\rho_{ij}|,
\end{aligned}
\end{equation}
where $|\rho_{ij}|$ are the magnitudes of the off-diagonal entries of the density matrix $\rho\in\mathcal{H}$.

In particular, for quantum coherence $C_{l_{1}}(\rho_{AB})$ of bipartite states $\rho_{AB}$, the relation (\ref{1}) refers to the following superadditivity one,
\begin{equation}
\begin{aligned}
C_{l_{1}}(\rho_{AB})\geqslant C_{l_{1}}(\rho_{A})+C_{l_{1}}(\rho_{B}).\label{bi1}
\end{aligned}
\end{equation}
The superadditivity inequalities of $C_{l_{1}}$ have been extended to multiqubit systems in Ref.\citen{S25},
\begin{equation}
\begin{aligned}
C_{l_{1}}(\rho_{A_{1}A_{2}\cdots A_{N}})\geqslant C_{l_{1}}(\rho_{A_{1}})+C_{l_{1}}(\rho_{A_{2}})+\cdots+C_{l_{1}}(\rho_{A_{N}}).\label{n1}
\end{aligned}
\end{equation}
For the $\alpha$-th ($\alpha\geqslant 1$) power of $C_{l_1}$, superadditivity relations in $N$-qubit systems has been presented \cite{S29},
\begin{equation}
\begin{split}
C_{l_{1}}^{\alpha}(\rho_{A_{1}A_{2}\cdots A_{N}})
&\geqslant \sum_{n=1}^{N} \lambda_{n}C_{l_{1}}^{\alpha}(\rho_{A_{n}}),\label{41}
\end{split}
\end{equation}
with $\lambda_{n}=1$ for $n\in{\{1,2, \cdots, N}\}$.
Further, strong superadditivity inequality (\ref{41}) holds for $N \geqslant 3$ with
$$
\lambda_{n}= \begin{cases}
(2^{\alpha}-1)^{n-1},\quad &n=1,2, \cdots ,m \\
(2^{\alpha}-1)^{m+1},\quad &n=m+1, \cdots ,N-1\\
(2^{\alpha}-1)^{m},\quad &n=N,
\end{cases}
$$
 where $m$ is one positive integer ($1\leqslant m\leqslant N-2$). It is valid under the conditions that $C_{l_{1}}(\rho_{A_{i}})\geqslant C_{l_{1}}(\rho_{A_{i+1}\cdots A_{N}})$ for $i=1, \cdots ,m$ and $C_{l_{1}}(\rho_{A_{j}})\leqslant C_{l_{1}}(\rho_{A_{j+1}\cdots A_{N}})$ for $j=m+1, \cdots ,N-1$ \cite{S29}. Later, the relation (\ref{41}) has been further improved by replacing $2^{\alpha}-1 $ in the coefficients $\lambda_{n}$ with $\frac{(1+k)^{\alpha}-1}{k^{\alpha}}$, under the conditions  $0<k\leqslant 1$, $kC_{l_{1}}(\rho_{A_{i}})\geqslant C_{l_{1}}(\rho_{A_{i+1}\cdots A_{N}})$ for $i=1, 2,\cdots ,m$, and $kC_{l_{1}}(\rho_{A_{j}})\leqslant C_{l_{1}}(\rho_{A_{j+1}\cdots A_{N}})$ for $j=m+1, \cdots, N-1$ \cite{S30}. In Ref.\citen{S31}, by introducing another real parameter $\delta$ ($\delta\geqslant 1$) and using $k^{\delta}$ instead of $k$, a class of further improved superadditivity relations is derived. Namely, the relation (\ref{41}) holds for
$$
\lambda_{n}= \begin{cases}
\left(\frac{(1+k^{\delta})^{\alpha}-1}{k^{\delta\alpha}}\right)^{n-1},\quad &n=1,2, \cdots ,m. \\
\left(\frac{(1+k^{\delta})^{\alpha}-1}{k^{\delta\alpha}}\right)^{m+1},\quad &n=m+1, \cdots ,N-1.\\
\left(\frac{(1+k^{\delta})^{\alpha}-1}{k^{\delta\alpha}}\right)^{m},\quad &n=N.
\end{cases}
$$
The above superadditivity relation depends on the inequality \cite{S31}:
\begin{equation}\label{ryy}
\begin{aligned}
(1+t)^{\alpha}\geqslant 1+\frac{(1+k^{\delta})^{\alpha}-1}{k^{\delta\alpha}}t^{\alpha}, 0\leqslant t\leqslant 1, \alpha\geqslant 1.
\end{aligned}
\end{equation}
To derive tighter superadditivity relations of $l_{1}$-norm coherence for multiqubit systems, we first prove the following lemmas.

\begin{lemma}
Let $x$ and $\alpha$ be real numbers satisfying $0\leqslant x\leqslant 1$ and $\alpha\geqslant 2$. Then we have
\begin{equation}\label{lemma1}
\begin{aligned}
(1+x)^{\alpha-1}\geqslant1+(\alpha-1)x.
\end{aligned}
\end{equation}
\end{lemma}

\begin{proof}
Set $h(x,\alpha)=(1+x)^{\alpha-1}-(\alpha-1)x-1$ with $0\leqslant x\leqslant 1$ and $\alpha\geqslant 2$. Since $\frac{\partial h(x,\alpha)}{\partial x}=(\alpha-1)(1+x)^{\alpha-2}-(\alpha-1)=(\alpha-1)[(1+x)^{\alpha-2}-1]\geqslant0$, the function $h(x,\alpha)$ is increasing with respect to $x$. As $0\leqslant x \leqslant 1$, we obtain the inequality (\ref{lemma1}).
\end{proof}

\begin{lemma}
Let $k$, $\delta$ and $\alpha$ be real numbers satisfying $0<k\leqslant 1$, $\delta\geqslant 1$ and $\alpha\geqslant 2$. For any $x$ satisfying $0\leqslant x \leqslant k^{\delta}$, we have
\begin{equation}\label{lemma2}
\begin{aligned}
(1+x)^{\alpha}\geqslant1+x+\frac{(1+k^{\delta})^{\alpha}-k^{\delta}-1}{k^{\delta\alpha}}x^{\alpha}.
\end{aligned}
\end{equation}
\end{lemma}

\begin{proof}
First we note the above inequality is trivial for $x=0$. So we prove the case for $x\neq0$. Consider the function $f(x,\alpha)=\frac{(1+x)^{\alpha}-x-1}{x^{\alpha}}$ with $0<x \leqslant k^{\delta}$, $0<k\leqslant 1$, $\delta\geqslant 1$ and $\alpha\geqslant 2$. Then by Lemma 1 we have
\begin{equation}
\begin{split}
\frac{\partial f(x,\alpha)}{\partial x}&=\frac{[\alpha(1+x)^{\alpha-1}-1]x^{\alpha}-\alpha x^{\alpha-1}[(1+x)^{\alpha}-x-1]}{x^{2\alpha}}\\
&=\frac{x^{\alpha-1}[-\alpha(1+x)^{\alpha-1}+(\alpha-1)x+\alpha]}{x^{2\alpha}}\leqslant 0,
\nonumber
\end{split}
\end{equation}
since $-\alpha(1+x)^{\alpha-1}+(\alpha-1)x+\alpha \leqslant 0$ when $\alpha\geqslant 2$. Therefore, $f(x,\alpha)$ is a decreasing function of $x$. As $0<x \leqslant k^{\delta}$, we obtain $f(x,\alpha)\geqslant f(k^{\delta},\alpha)=\frac{(1+k^{\delta})^{\alpha}-k^{\delta}-1}{k^{\delta\alpha}}$.
\end{proof}

\section{Main Results}

\subsection{Superadditivity of the $\alpha$-th power of multipartite $l_{1}$-norm coherence}

Based on these inequalities in lemmas, we derive the following superadditivity inequality between the total coherence and local coherence for an arbitrary multiqubit state $\rho_{A_{1},\dots,A_{N}}$. For simplicity we denote $\mathcal{C}_{1,2,\dots ,N}$ the coherence $C_{l_{1}}(\rho_{A_{1},\dots,A_{N}})$ and $\mathcal{C}_{n}$ the coherence $C_{l_{1}}(\rho_{A_{n}})$.

\begin{theorem}
Suppose $k$ and $\delta$ are real numbers satisfying $0<k\leqslant 1$ and $\delta\geqslant 1$. For the $\alpha$-th ($\alpha\geqslant 2$) power of the $l_{1}$-norm coherence in $N$-qubit system ($N\geqslant 3$), the following superadditivity inequality holds
\begin{equation}
	\begin{split}
	\mathcal{C}_{1,2,\dots ,N}^{\alpha}
&\geqslant \sum_{n=1}^{N} \Omega_{n}\Gamma^{n-1}\mathcal{C}_{n}^{\alpha},
\label{the1}
	\end{split}
\end{equation}
with $\Gamma =\frac{(1+k^{\delta})^{\alpha}-k^{\delta}-1}{k^{\delta\alpha}}$ and $\Omega_{n}=1+\frac{\mathcal{C}_{n+1,\dots ,N}}{\mathcal{C}_{n}}$, $n=1,2, \cdots, N-1$, $\Omega_{N}=1$,
if $k^{\delta}\mathcal{C}_{n}\geqslant \mathcal{C}_{n+1,\dots ,N}$ for $n=1,2, \cdots, N-1$.
\end{theorem}

\begin{proof}
From (\ref{bi1}) and (\ref{lemma2}) we have
\begin{equation}
\begin{split}
\mathcal{C}_{1,2,\dots,N}^{\alpha}
&\geqslant
(\mathcal{C}_{1}+\mathcal{C}_{2,\dots,N})^{\alpha}\\
&=\mathcal{C}_{1}^{\alpha}(1+\frac{\mathcal{C}_{2,\cdots,N}}{\mathcal{C}_{1}})^{\alpha}\\
&\geqslant \mathcal{C}_{1}^{\alpha}\left[1+\frac{\mathcal{C}_{2,\cdots,N}}
{\mathcal{C}_{1}}+\frac{(1+k^{\delta})^{\alpha}-1-k^{\delta}}
{k^{\delta\alpha}}\frac{\mathcal{C}_{2,\cdots,N}^{\alpha}}{\mathcal{C}_{1}^{\alpha}}\right]\\
&=\left(1+\frac{\mathcal{C}_{2,\cdots,N}}{\mathcal{C}_{1}}\right)
\mathcal{C}_{1}^{\alpha}+\frac{(1+k^{\delta})^{\alpha}-1-k^{\delta}}
{k^{\delta\alpha}}\mathcal{C}_{2,\cdots,N}^{\alpha}\\
&\geqslant \cdots \\
&\geqslant \left(1+\frac{\mathcal{C}_{2,\cdots,N}}{\mathcal{C}_{1}}\right)
\mathcal{C}_{1}^{\alpha}+\left(1+\frac{\mathcal{C}_{3,\cdots,N}}
{\mathcal{C}_{2}}\right)\frac{(1+k^{\delta})^{\alpha}-1-k^{\delta}}
{k^{\delta\alpha}}\mathcal{C}_{2}^{\alpha}+\cdots\\
&\quad +\left(1+\frac{\mathcal{C}_{N}}{\mathcal{C}_{N-1}}\right)
\left(\frac{(1+k^{\delta})^{\alpha}-1-k^{\delta}}{k^{\delta\alpha}}\right)^{N-2}
\mathcal{C}_{N-1}^{\alpha}+\left(\frac{(1+k^{\delta})^{\alpha}-1-k^{\delta}}
{k^{\delta\alpha}}\right)^{N-1}\mathcal{C}_{N}^{\alpha}\\
\label{thep1}
\end{split}
\end{equation}
for all $\alpha\geqslant2$.
\end{proof}

\begin{remark}
For the superadditivity relation of the $l_{1}$-norm coherence in this paper, it is assumed that $\mathcal{C}_{n}\neq0$ for all $n$. If there exists $\mathcal{C}_{i}=0$, while $\mathcal{C}_{j}\neq0$ for all $j\neq i$, the corresponding term $\mathcal{C}_{i}^{\alpha}$
in the lower bound in Theorem 1 will be removed, together with the subsequent terms with the coefficients multiplied by $\frac{1}{\Gamma}$, as manifested in the proof of the theorem by using the inequalities (\ref{bi1}) and (\ref{lemma2}).
\end{remark}
Theorem 1 gives a new class of inequality relations between the $\alpha$-th power of the multipartite $l_{1}$-norm coherence and the nonlinear combinations of each single qubit coherence under the conditions $k^{\delta}\mathcal{C}_{n}\geqslant \mathcal{C}_{n+1,\dots ,N}$ for $n=1,2, \cdots, N-1$. Note that the coefficients $\Omega_{n}$ depend on $\mathcal{C}_{n}$. The coefficients $\Omega_{n}\Gamma^{n-1}$ in inequality (\ref{the1}) are greater than or equal to 1 for given quantum states and given parameters $k$ and $\delta$. Since generally
\begin{equation}
	\begin{split}
x+\frac{(1+k^{\delta})^{\alpha}-k^{\delta}-1}{k^{\delta\alpha}}x^{\alpha}
&=\frac{(1+k^{\delta})^{\alpha}-1}{k^{\delta\alpha}}x^{\alpha}+k^{\delta}
[\frac{x}{k^{\delta}}-(\frac{x}{k^{\delta}})^{\alpha}]\\
&\geqslant \frac{(1+k^{\delta})^{\alpha}-1}{k^{\delta\alpha}}x^{\alpha}
\nonumber
\end{split}
\end{equation}
for $0<x \leqslant k^{\delta}\leqslant k\leqslant 1$ and $\alpha \geqslant 2$, where the inequality is due to that $\frac{x}{k^{\delta}}-(\frac{x}{k^{\delta}})^{\alpha}\geqslant 0$,
our lower bounds of inequality (\ref{lemma2}) are larger than the existing ones given in Refs.\citen{S29,S30,S31}.

For arbitrary $N$-qubit states, the conditions in Theorem 1 may not be satisfied. The following theorem gives superadditivity inequalities under other scenarios.

\begin{theorem}
If for some positive integer $m$ ($1\leqslant m\leqslant N-2$), $k^{\delta}\mathcal{C}_{n}\geqslant \mathcal{C}_{n+1,\dots ,N}$ for $n=1,2, \cdots, m$, and $\mathcal{C}_{\dot{n}}\leqslant k^{\delta}\mathcal{C}_{\dot{n}+1,\dots ,N}$ for $\dot{n}=m+1, \cdots, N-1$, then we have
\begin{equation}
	\begin{split}
	\mathcal{C}_{1,2,\dots ,N}^{\alpha}
&\geqslant \sum_{n=1}^{m} \Gamma^{n-1}\Omega_{n}\mathcal{C}_{n}^{\alpha}+\Gamma^{m+1}
\mathcal{C}_{m+1}^{\alpha}+\Gamma^{m+1}\sum_{n=m+2}^{N-1}\Upsilon_{m+1}\cdots \Upsilon_{n-1}\mathcal{C}_{n}^{\alpha}+\Gamma^{m}\Upsilon_{m+1}\cdots \Upsilon_{N-1}\mathcal{C}_{N}^{\alpha}
\label{the2}
	\end{split}
\end{equation}
for $\alpha\geqslant 2$, $N\geqslant 3$, $0<k\leqslant 1$ and $\delta\geqslant 1$, where $\Gamma =\frac{(1+k^{\delta})^{\alpha}-k^{\delta}-1}{k^{\delta\alpha}}$, $\Omega_{n}=1+\frac{\mathcal{C}_{n+1,\dots ,N}}{\mathcal{C}_{n}}$ for $n=1,2, \cdots, m$ and $\Upsilon_{\dot{n}}=1+\frac{\mathcal{C}_{\dot{n}}}{\mathcal{C}_{\dot{n}+1,\dots ,N}}$ for $\dot{n}=m+1,\cdots, N-1$.
\end{theorem}

\begin{proof}
According to Theorem 1, we have
\begin{equation}
\begin{split}
\mathcal{C}_{1,2,\dots,N}^{\alpha}
&\geqslant \left(1+\frac{\mathcal{C}_{2,\cdots,N}}{\mathcal{C}_{1}}\right)
\mathcal{C}_{1}^{\alpha}+\left(1+\frac{\mathcal{C}_{3,\cdots,N}}
{\mathcal{C}_{2}}\right)\frac{(1+k^{\delta})^{\alpha}-1-k^{\delta}}
{k^{\delta\alpha}}\mathcal{C}_{2}^{\alpha}+\cdots+
\left(1+\frac{\mathcal{C}_{m+1,\cdots,N}}{\mathcal{C}_{m}}\right)\\
&\quad \times\left(\frac{(1+k^{\delta})^{\alpha}-1-k^{\delta}}
{k^{\delta\alpha}}\right)^{m-1}\mathcal{C}_{m}^{\alpha}
+\left(\frac{(1+k^{\delta})^{\alpha}-1-k^{\delta}}
{k^{\delta\alpha}}\right)^{m}\mathcal{C}_{m+1,\cdots,N}^{\alpha}\\
&=\Omega_{1}\mathcal{C}_{1}^{\alpha}+\Omega_{2}
\Gamma\mathcal{C}_{2}^{\alpha}+\cdots+\Omega_{m}\Gamma^{m-1}
\mathcal{C}_{m}^{\alpha}+\Gamma^{m}\mathcal{C}_{m+1,\cdots,N}^{\alpha}.
\label{thep2}
\end{split}
\end{equation}
Similarly, since $\mathcal{C}_{n}\leqslant k^{\delta}\mathcal{C}_{n+1,\dots ,N}$ for $n=m+1, \cdots, N-1$, we have
\begin{equation}
\begin{split}
\mathcal{C}_{m+1,\cdots,N}^{\alpha}&\geqslant \left(1+\frac{\mathcal{C}_{m+1}}{\mathcal{C}_{m+2,\cdots,N}}\right)
\mathcal{C}_{m+2,\cdots,N}^{\alpha}+\frac{(1+k^{\delta})^{\alpha}-k^{\delta}-1}
{k^{\delta\alpha}}\mathcal{C}_{m+1}^{\alpha}\\
&=\Upsilon_{m+1}\mathcal{C}_{m+2,\cdots,N}^{\alpha}+\Gamma\mathcal{C}_{m+1}^{\alpha}\\
&\geqslant (\Upsilon_{m+1}\cdots\Upsilon_{N-1})\mathcal{C}_{N}^{\alpha}
+\Gamma[(\Upsilon_{m+2}\cdots\Upsilon_{N-1})\mathcal{C}_{N-1}^{\alpha}+\cdots
+\Upsilon_{m+1}\mathcal{C}_{m+2}^{\alpha}+\mathcal{C}_{m+1}^{\alpha}].
\label{the2p3}
\end{split}
\end{equation}
Substituting the inequality (\ref{the2p3}) into (\ref{thep2}), we complete the proof of Theorem 2.
\end{proof}

Theorem 2 provides more comprehensive superadditivity inequalities as it encompasses a broader range of conditions compared to Theorem 1. It is worth emphasizing that Theorems 1 and 2 give only two special classes of superadditivity relations. More generally, for $N$-qubit systems there could be $2^{N-1}$ similar inequalities. Closer relations can be similarly obtained by adopting our approach. For example, for three-qubit states we also have the following conclusion.

\begin{corollary}
For any 3-qubit state $\rho_{A_{1}A_{2}A_{3}}$ satisfying the conditions $\mathcal{C}_{1}\leqslant k^{\delta}\mathcal{C}_{23}$ and $k^{\delta}\mathcal{C}_{2}\geqslant \mathcal{C}_{3}$ ($0<k\leqslant 1$ and $\delta\geqslant 1$), we have
\begin{equation}
\mathcal{C}_{123}^{\alpha}
\geqslant (1+\frac{\mathcal{C}_{1}}{\mathcal{C}_{23}})(1+\frac{\mathcal{C}_{3}}
{\mathcal{C}_{2}})\mathcal{C}_{2}^{\alpha}+(1+\frac{\mathcal{C}_{1}}
{\mathcal{C}_{23}})\frac{(1+k^{\delta})^{\alpha}-1-k^{\delta}}{k^{\delta\alpha}}
\mathcal{C}_{3}^{\alpha}+\frac{(1+k^{\delta})^{\alpha}-1-k^{\delta}}{k^{\delta\alpha}}
\mathcal{C}_{1}^{\alpha}
\label{co}
\end{equation}
for all $\alpha\geqslant2$.
\end{corollary}

\begin{proof}
From the proof of Theorem 2, we have
\begin{equation}
\begin{split}
\mathcal{C}_{123}^{\alpha}&\geqslant
\left(\mathcal{C}_{1}+\mathcal{C}_{23}\right)^{\alpha}\\
&=\mathcal{C}_{23}^{\alpha}(1+\frac{\mathcal{C}_{1}}{\mathcal{C}_{23}})^{\alpha}\\
&\geqslant \mathcal{C}_{23}^{\alpha}\left[1+\frac{\mathcal{C}_{1}}{\mathcal{C}_{23}}
+\frac{(1+k^{\delta})^{\alpha}-1-k^{\delta}}{k^{\delta\alpha}}
\frac{\mathcal{C}_{1}^{\alpha}}{\mathcal{C}_{23}^{\alpha}}\right]\\
&=\left(1+\frac{\mathcal{C}_{1}}{\mathcal{C}_{23}}\right)
\mathcal{C}_{23}^{\alpha}+\frac{(1+k^{\delta})^{\alpha}-1-k^{\delta}}
{k^{\delta\alpha}}\mathcal{C}_{1}^{\alpha}\\
&\geqslant \left(1+\frac{\mathcal{C}_{1}}{\mathcal{C}_{23}}\right)
\left(\mathcal{C}_{2}+\mathcal{C}_{3}\right)^{\alpha}
+\frac{(1+k^{\delta})^{\alpha}-1-k^{\delta}}{k^{\delta\alpha}}\mathcal{C}_{1}^{\alpha}\\
&\geqslant \left(1+\frac{\mathcal{C}_{1}}{\mathcal{C}_{23}}\right)
\mathcal{C}_{2}^{\alpha}\left[1+\frac{\mathcal{C}_{3}}
{\mathcal{C}_{2}}+\frac{(1+k^{\delta})^{\alpha}-1-k^{\delta}}{k^{\delta\alpha}}
\frac{\mathcal{C}_{3}^{\alpha}}{\mathcal{C}_{2}^{\alpha}}\right]
+\frac{(1+k^{\delta})^{\alpha}-1-k^{\delta}}{k^{\delta\alpha}}\mathcal{C}_{1}^{\alpha}\\
&=(1+\frac{\mathcal{C}_{1}}{\mathcal{C}_{23}})
(1+\frac{\mathcal{C}_{3}}{\mathcal{C}_{2}})\mathcal{C}_{2}^{\alpha}
+(1+\frac{\mathcal{C}_{1}}{\mathcal{C}_{23}})\frac{(1+k^{\delta})^{\alpha}
-1-k^{\delta}}{k^{\delta\alpha}}\mathcal{C}_{3}^{\alpha}
+\frac{(1+k^{\delta})^{\alpha}-1-k^{\delta}}{k^{\delta\alpha}}\mathcal{C}_{1}^{\alpha},
\label{thep1}
\end{split}
\end{equation}
which completes the proof.
\end{proof}

\begin{example}
Let us consider the three-qubit state:
$$
|\varphi_{A_{1}A_{2}A_{3}}\rangle=\frac{|0\rangle+|1\rangle}{\sqrt{2}}\otimes \frac{|0\rangle+2|1\rangle}{\sqrt{5}}\otimes \frac{|0\rangle+3|1\rangle}{\sqrt{10}}.
$$
Under the computational basis $\{|i\rangle\langle i|\}$, we have the $l_1$-norm coherence $\mathcal{C}_{1}=1$,
$\mathcal{C}_{2}=\frac{4}{5}$, $\mathcal{C}_{3}=\frac{3}{5}$, $\mathcal{C}_{23}=\frac{47}{25}$ and $\mathcal{C}_{123}=\frac{119}{25}$. The lower bound of $\mathcal{C}_{123}^{\alpha}$ with respect to $\alpha$ and $k$ is shown in fig.\ref{f0}(\textbf{a}) for $\delta=2$. The bounds of $\mathcal{C}_{123}^{\alpha}$ with respect to $\alpha$ is plotted in fig.\ref{f0}(\textbf{b}) for fixed $k=\frac{9}{10}$. The corresponding superadditivity lower bound in Ref.\citen{S31} is given by
$$
\mathcal{C}_{123}^{\alpha}
\geqslant \mathcal{C}_{2}^{\alpha}+\frac{(1+k^{\delta})^{\alpha}-1}{k^{\delta\alpha}}
\mathcal{C}_{3}^{\alpha}+\frac{(1+k^{\delta})^{\alpha}-1}{k^{\delta\alpha}}
\mathcal{C}_{1}^{\alpha}.
$$
Obviously our bound in (\ref{co}) is indeed tighter than the ones given in Refs.\citen{S29,S30,S31}, see fig.\ref{f0}.
\end{example}

\begin{figure}[h]\nonumber
\centering
\raisebox{12em}{(\textbf{a})}\includegraphics[width=5.25cm]{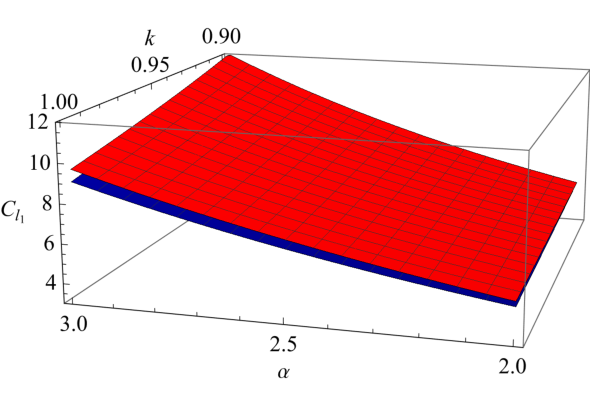}\quad
\raisebox{12em}{(\textbf{b})}\includegraphics[width=5.25cm]{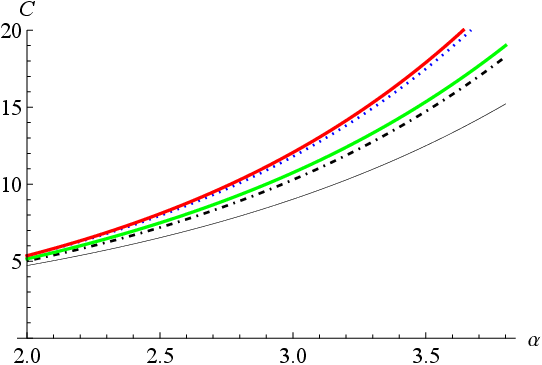}
\end{figure}
\begin{figure}[h]
\begin{center}
\caption{The super-additivity of the $\alpha$-th ($\alpha\geqslant 2$) power of $l_1$-norm coherence for $|\varphi_{123}\rangle$. \textbf{a}. The lower bound for $C_{l_1}^{\alpha}$ with respect to $\alpha$ and $k$: the red surface represents the lower bound in (\ref{co}) and the bottom blue one is from Ref.~\citen{S31}. \textbf{b}. Bounds of $C_{123}^{\alpha}$ with respect to $\alpha$ for $k=0.9$: The red thick and blue dotted curves represent our bound and the lower bound in Ref.~\citen{S31} when $k = 0.9$ and $\delta =2$, respectively. The green thick and the black dash-dotted curves represent the lower bounds in (\ref{co}) and in Ref.~\citen{S30} when $k = 0.9$ and $\delta =1$, respectively. The black thin line is the bound from Ref.~\citen{S29}.}\label{f0}
\end{center}
\end{figure}

\subsection{Improved superadditivity of the $\alpha$-th power of multipartite $l_{1}$-norm coherence}

By choosing different $k_{n}$ ($k_{n}\leqslant k$), new improved superadditivity inequalities can be derived.

\begin{theorem}
For the $\alpha$-th ($\alpha\geqslant 2$) power of $l_{1}$-norm coherence
in $N$-qubit systems ($N\geqslant 3$), the following superadditivity inequality holds,
\begin{equation}
	\begin{split}
	\mathcal{C}_{1,2,\dots ,N}^{\alpha}
&\geqslant \sum_{n=1}^{N} \Omega_{n}\Gamma_{0}\cdots\Gamma_{n-1}\mathcal{C}_{n}^{\alpha},
\label{the3}
	\end{split}
\end{equation}
with $\Gamma_{n} =\frac{(1+k_{n}^{\delta})^{\alpha}-k_{n}^{\delta}-1}{k_{n}^{\delta\alpha}}$, $n=1,2, \cdots, N$, $\Gamma_{0}=1$ and $\Omega_{n}=1+\frac{\mathcal{C}_{n+1,\dots ,N}}{\mathcal{C}_{n}}$, $n=1,2, \cdots, N-1$, $\Omega_{N}=1$,
if $k_{n}^{\delta}\mathcal{C}_{n}\geqslant \mathcal{C}_{n+1,\dots ,N}$, $0<k_{n}\leqslant 1$ for $n=1,2, \cdots, N-1$ and $\delta\geqslant 1$ .
\end{theorem}

\begin{proof}
Similar to the proof of Theorem 1 we have
\begin{equation}
\begin{split}
\mathcal{C}_{1,2,\dots,N}^{\alpha}
&\geqslant
(\mathcal{C}_{1}+\mathcal{C}_{2,\dots,N})^{\alpha}\\
&=\mathcal{C}_{1}^{\alpha}(1+\frac{\mathcal{C}_{2,\cdots,N}}{\mathcal{C}_{1}})^{\alpha}\\
&\geqslant \mathcal{C}_{1}^{\alpha}\left[1+\frac{\mathcal{C}_{2,\cdots,N}}{\mathcal{C}_{1}}
+\frac{(1+k_{1}^{\delta})^{\alpha}-1-k_{1}^{\delta}}{k_{1}^{\delta\alpha}}
\frac{\mathcal{C}_{2,\cdots,N}^{\alpha}}{\mathcal{C}_{1}^{\alpha}}\right]\\
&=\left(1+\frac{\mathcal{C}_{2,\cdots,N}}{\mathcal{C}_{1}}\right)\mathcal{C}_{1}^{\alpha}
+\frac{(1+k_{1}^{\delta})^{\alpha}-1-k_{1}^{\delta}}{k_{1}^{\delta\alpha}}
\mathcal{C}_{2,\cdots,N}^{\alpha}\\
&\geqslant \left(1+\frac{\mathcal{C}_{2,\cdots,N}}{\mathcal{C}_{1}}\right)
\mathcal{C}_{1}^{\alpha}+\left(1+\frac{\mathcal{C}_{3,\cdots,N}}
{\mathcal{C}_{2}}\right)\frac{(1+k_{1}^{\delta})^{\alpha}-1-k_{1}^{\delta}}
{k_{1}^{\delta\alpha}}\mathcal{C}_{2}^{\alpha}+\cdots\\
&\quad+\left(1+\frac{\mathcal{C}_{N}}{\mathcal{C}_{N-1}}\right)
\left(\frac{(1+k_{1}^{\delta})^{\alpha}-1-k_{1}^{\delta}}{k_{1}^{\delta\alpha}}\right)\cdots \left(\frac{(1+k_{N-2}^{\delta})^{\alpha}-1-k_{N-2}^{\delta}}{k_{N-2}^{\delta\alpha}}\right)
\mathcal{C}_{N-1}^{\alpha}\\
&\quad+\left(\frac{(1+k_{1}^{\delta})^{\alpha}-1-k_{1}^{\delta}}{k_{1}^{\delta\alpha}}\right)\cdots \left(\frac{(1+k_{N-1}^{\delta})^{\alpha}-1-k_{N-1}^{\delta}}
{k_{N-1}^{\delta\alpha}}\right)\mathcal{C}_{N}^{\alpha}
\label{thep3}
\end{split}
\end{equation}
for all $\alpha\geqslant2$. This completes the proof.
\end{proof}

Based on Theorem 2, the following theorem is obtained by using the same method as above.

\begin{theorem}
If there exists positive integer $m$ ($1\leqslant m\leqslant N-2$)
such that $k_{n}^{\delta}\mathcal{C}_{n}\geqslant \mathcal{C}_{n+1,\dots ,N}$ for $n=1,2, \cdots, m$, and $\mathcal{C}_{\dot{n}}\leqslant k^{\delta}_{\dot{n}}\mathcal{C}_{\dot{n}+1,\dots ,N}$ for $\dot{n}=m+1, \cdots, N-1$, we have
\begin{equation}
	\begin{split}
	\mathcal{C}_{1,2,\dots ,N}^{\alpha}
&\geqslant \sum_{n=1}^{m} \Gamma_{0}\cdots\Gamma_{n-1}\Omega_{n}\mathcal{C}_{n}^{\alpha}
+(\Gamma_{0}\cdots\Gamma_{m+1})\mathcal{C}_{m+1}^{\alpha}+\Gamma_{0}\cdots\Gamma_{m+1}\\
&\quad \times\left(\sum_{n=m+2}^{N-1}\Upsilon_{m+1}\cdots \Upsilon_{n-1}\Gamma_{n}\mathcal{C}_{n}^{\alpha}+\Upsilon_{m+1}\cdots \Upsilon_{N-1}\mathcal{C}_{n}^{\alpha}\right)
\label{the4}
\end{split}
\end{equation}
for $\alpha\geqslant 2$, $N\geqslant 3$, $0<k_{n}\leqslant 1$, $\delta\geqslant 1$ and $\Gamma_{n} =\frac{(1+k_{n}^{\delta})^{\alpha}-k_{n}^{\delta}-1}{k_{n}^{\delta\alpha}}$, $n=1,2, \cdots, N$, $\Gamma_{0}=1$, $\Omega_{n}=1+\frac{\mathcal{C}_{n+1,\dots ,N}}{\mathcal{C}_{n}}$, $n=1,2, \cdots, m$,
and $\Upsilon_{\dot{n}}=1+\frac{\mathcal{C}_{\dot{n}}}{\mathcal{C}_{\dot{n}+1,\dots ,N}}$, $\dot{n}=m+1,\cdots, N-1$.
\end{theorem}

Notice that in Theorems 1 and 2 the common parameter $k$ has been taken. Due to the fact that the function $f(x,\alpha)=\frac{(1+x)^{\alpha}-x-1}{x^{\alpha}}$ is decreasing with respect to $x$ ($0<x \leqslant k_{n}^{\delta}\leqslant k^{\delta}\leqslant 1$) for $\alpha\geqslant 2$, we have
$$
\Gamma_{n} =\frac{(1+k_{n}^{\delta})^{\alpha}-k_{n}^{\delta}-1}{k_{n}^{\delta\alpha}}\geqslant
\frac{(1+k^{\delta})^{\alpha}-k^{\delta}-1}{k^{\delta\alpha}}=\Gamma
$$
for $n\in{\{1,2, \cdots, N-1}\}$ and $\delta\geqslant 1$.
Therefore, the superadditivity inequalities in Theorems 3 and 4 are tighter than the ones in Theorems 1 and 2. We give another example below to illustrate the point.

\begin{example}
We consider the three-qubit pure state in Schmidt decomposition \cite{S32,S33},
\begin{eqnarray}
|\psi_{A_{1}A_{2}A_{3}} \rangle =\lambda _{0}|000\rangle +\lambda _{1}e^{i\phi }|100\rangle
+\lambda _{2}|101\rangle  \notag +\lambda _{3}|110\rangle +\lambda _{4}|111\rangle ,  \label{sd-1}
\end{eqnarray}
where $\lambda _{i}\geqslant 0$, $i=0,1,\cdots ,4$, $0\leqslant \phi <2\pi $ and $\sum \lambda _{i}^{2}=1$. Setting $\phi=0$ and $\lambda _{i}=\frac{1}{\sqrt{5}}$, we have
$\mathcal{C}_{1}=\frac{2}{5}$, $\mathcal{C}_{2}=\mathcal{C}_{3}=\frac{4}{5}$, $\mathcal{C}_{23}=\frac{12}{5}$ and $\mathcal{C}_{123}=\frac{18}{5}$.
Since $\mathcal{C}_{2}=\mathcal{C}_{3}$, we can only take $k_{2}$ to be 1 in this case.
When $\delta=2$, we obtain
%Notice that $\frac{\mathcal{C}_{1}}{\mathcal{C}_{23}}=\frac{1}{6}< 1$, we also add the situation in which $k$ takes different values in Ref.\cite{S31}.
\begin{equation*}
\begin{split}
\mathcal{S}_1 &\equiv \mathcal{C}_{123}^{\alpha}- (1+\frac{\mathcal{C}_{1}}{\mathcal{C}_{23}})(1+\frac{\mathcal{C}_{2}}
{\mathcal{C}_{3}})\mathcal{C}_{3}^{\alpha}-(1+\frac{\mathcal{C}_{1}}
{\mathcal{C}_{23}})\frac{(1+k_{2}^{\delta})^{\alpha}-1-k_{2}^{\delta}}
{k_{2}^{\delta\alpha}}\mathcal{C}_{2}^{\alpha}
-\frac{(1+k_{1}^{\delta})^{\alpha}-1-k_{1}^{\delta}}{k_{1}^{\delta\alpha}}
\mathcal{C}_{1}^{\alpha}\\
&=\left(\frac{18}{5}\right)^{\alpha}-\left(\frac{7}{3}\right)
\times\left(\frac{4}{5}\right)^{\alpha}-\left(\frac{7}{6}\right)
\times\left(2^{\alpha}-2\right)\times\left(\frac{4}{5}\right)^{\alpha}
-\frac{(1+k_{1}^{2})^{\alpha}-1-k_{1}^{2}}{k_{1}^{2\alpha}}
\times\left(\frac{2}{5}\right)^{\alpha}.
\nonumber
\end{split}
\end{equation*}
While from Ref.\citen{S31} one has
\begin{equation*}
\begin{split}
\mathcal{S}_2 &\equiv \mathcal{C}_{123}^{\alpha}-\mathcal{C}_{3}^{\alpha}-\frac{(1+k_{2}^{\delta})^{\alpha}-1}
{k_{2}^{\delta\alpha}}\mathcal{C}_{2}^{\alpha}-\frac{(1+k_{1}^{\delta})^{\alpha}-1}
{k_{1}^{\delta\alpha}}\mathcal{C}_{1}^{\alpha}\\
&=\left(\frac{18}{5}\right)^{\alpha}-\left(\frac{4}{5}\right)^{\alpha}
-\left(2^{\alpha}-1\right)\times\left(\frac{4}{5}\right)^{\alpha}
-\frac{(1+k_{1}^{2})^{\alpha}-1}{k_{1}^{2\alpha}}\times\left(\frac{2}{5}\right)^{\alpha}.
\nonumber
\end{split}
\end{equation*}
\end{example}
Fig. \ref{f1}(\textbf{a}) shows that our superadditivity inequality is tighter.
When $k_{1}^{\delta}=\frac{1}{6}$ and $k_{2}=1$, the lower bound of our superadditivity inequality is tighter than that given by Ref.\citen{S31}. As shown in fig. \ref{f1}(\textbf{b}), Theorem 4 gives a lower bound that is significantly tighter than the bound in Theorem 2 when $k_{1}=k_{2}=1$. When $k$ takes the same value 1, the result in Ref.\citen{S31} is exactly the lower bound given in Ref.\citen{S29}. Moreover, our lower bound is strictly larger than the bound in Ref.\citen{S29} for $k_{1}=k_{2}=1$. Thus it is verified that our results are all stronger than the ones given in Ref.\citen{S31}, whenever the $k_n$ is the same or not.

\begin{figure}[h]\nonumber
\centering
\raisebox{12em}{(\textbf{a})}\includegraphics[width=5.25cm]{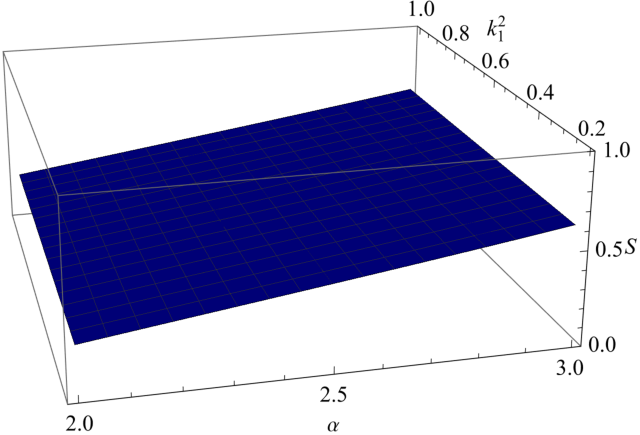}\quad\quad\quad
\raisebox{12em}{(\textbf{b})}\includegraphics[width=5.25cm]{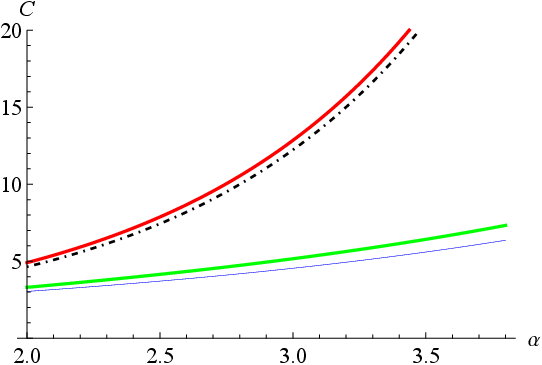}
\end{figure}
\begin{figure}[h]
\begin{center}
\caption{The superadditivity of the $\alpha$-th ($\alpha\geqslant 2$) power of $l_1$-norm coherence for $|\psi_{123} \rangle$. \textbf{a}. The blue surface represents $\mathcal{S}_2-\mathcal{S}_1>0$ for the states $|\psi_{123} \rangle$ with respect to $\alpha$ and $k_{1}$. \textbf{b}. Bounds for $l_1$-norm coherence $C_{l_1}^{\alpha}$ for $k_{2}=1$ and $\delta =1$: the red thick and the black dash-dotted curves represent our bound and the bound from Ref.\citen{S31}, respectively, when $k_{1} = \frac{1}{6}$. The blue thin curve represents the lower bound when $k_{1} = 1$ in Ref.\citen{S31}. The green thick line is the bound when $k_{1}=k_{2} = 1$ in Theorem 2.}\label{f1}
\end{center}
\end{figure}

\section{Conclusion}
In this paper, the multipartite coherence distribution is described as the superadditivity relations between the total coherence and the local coherence of single qubits, quantified by the $l_1$-norm coherence measure $C_{l_1}$. We have derived a set of stronger superadditivity inequalities related to the $C_{l_1}^{\alpha}$ ($\alpha\geqslant 2$) for multiqubit systems. We have different general forms of the inequalities, depending on the types of conditions that the considered state satisfies. These inequalities can be used to characterize the coherence distributions among the multiqubit systems. It is interesting to note that the authors in Refs.\citen{S34,S35,S36} have demonstrated that the correlated coherence $\mathcal{S}=\mathcal{C}_{123}-\mathcal{C}_{1}-\mathcal{C}_{2}-\mathcal{C}_{3}$, as a new coherence measure for three-qubit systems, has monogamy properties. In this sense, our results may shed new lights on exploring other forms of coherence distribution relations.

\section*{Acknowledgments}

This work was supported by the NSFC (grant nos. 12075159 and 12171044), the specific research fund of the Innovation Platform for Academician of Hainan Province under Grant No. YSPTZX202215.

\end{document}